\documentclass[12pt]{article}
\usepackage{amscd,amsmath,amssymb,amstext,amsthm,exscale,latexsym}
\usepackage{graphicx,floatflt}
\newcommand {\al}   {\alpha}       \newcommand {\bt}  {\beta}
\newcommand {\g }   {\gamma}       
\newcommand {\dl}   {\delta}       \newcommand {\e }  {\epsilon}
\newcommand {\z }   {\zeta}        
\newcommand {\ve}   {\varepsilon}

\newcommand {\vf }  {\varphi}      
         
\newcommand {\Lm}   {\Lambda}      \newcommand {\Om}  {\Omega}
       
\newcommand {\pl}   {\partial}     \newcommand {\nb}  {\nabla}

\newcommand   {\const}{{\sf\,const}}     \newcommand   {\diag}{{\sf\,diag\,}}
         \newcommand   {\Lie}{{\sf\,L}}
            
\newcommand   {\osmall}{{\sf o}}        

\newcommand {\MG}  {{\mathbb G}}   \newcommand {\MH}  {{\mathbb H}}
\newcommand {\MI}  {{\mathbb I}}   
   
\newcommand {\MM}  {{\mathbb M}}   
\newcommand {\MO}  {{\mathbb O}}   
   \newcommand {\MR}  {{\mathbb R}}
\newcommand {\MS}  {{\mathbb S}}   
\newcommand {\MU}  {{\mathbb U}}

      \newcommand {\CX}  {{\cal X}}

\newcommand {\Sn}  {{\textsc{n}}}

\newcommand {\CC }  {{\cal C}}

\newcommand {\Gi}  {\mathfrak{i}}

\newtheorem{prop}{Proposition}[section]
\newtheorem{exa}{Example}[section]
\newtheorem*{defn}{Definition}
\newtheorem{theorem}{Theorem}[section]
\newtheorem*{cor}{Corollary}
\begin{document}
\title     {Killing vector fields and a homogeneous isotropic universe}
\author    {M.~O.~Katanaev
            \thanks{E-mail: katanaev@mi.ras.ru}\\ \\
            \sl Steklov Mathematical Institute,\\
            \sl ul.~Gubkina, 8, Moscow, 119991, Russia}
\date      {20 September 2016}
\maketitle
\begin{abstract}
Some basic theorems on Killing vector fields are reviewed. In particular, the
topic of a constant-curvature space is examined. A detailed proof is given for a
theorem describing the most general form of the metric of a homogeneous
isotropic space-time. Although this theorem can be considered to be commonly
known, its complete proof is difficult to find in the literature. An example
metric is presented such that all its spatial cross sections correspond to
constant-curvature spaces, but it is not homogeneous and isotropic as a whole.
An equivalent definition of a homogeneous and isotropic space-time in terms of
embedded manifolds is also given.
\end{abstract}
\section{Introduction}
We first introduce our conventions and definitions, and then formulate the
theorem.
\begin{defn}
A space-time is a pair $(\MM,g)$, where $\MM$ is a four dimensional
differentiable manifold and $g$ is a metric of the Lorentzian signature
$(+---)$ given on it.
\end{defn}
We assume that both the manifold $\MM$ and the metric $g$ are sufficiently
smooth. Furthermore, we assume that the manifold is geodesically complete, i.e.,
any geodesic can be extended in both directions beyond any value of the
canonical parameter. In general relativity, a space-time manifold is commonly
geodesically incomplete due to singularities of solutions of the Einstein
equations. In this paper, we do no consider dynamical equations, focusing on the
kinematical aspect only. Therefore, our assumption on the geodesic completeness
is quite natural and, for example, does not allow considering only a part of the
whole sphere.

If the null coordinate line in local coordinates $x^\al$, $\al=0,1,2,3$, is
time-like, $(\pl_0,\pl_0)=g_{00}>0$, where the parenthesis denote the scalar
product, then the coordinate $x^0:=t$ is called time. Spatial indices are
denoted by Greek letters from the middle of the alphabet:
$\mu,\nu,\dotsc=1,2,3$. Then $\lbrace x^\al\rbrace=\lbrace x^0,x^\mu\rbrace$.

The modern observational data indicate that our Universe is homogeneous and
isotropic (the cosmological principle), at least at the first approximation.
Most cosmological models rely on the following statement.
\begin{theorem}                                                   \label{tbchyt}
Let a four-dimensional manifold be the topological product $\MM=\MR\times\MS$,
where $t\in\MR$ is the time coordinate and $x\in\MS$ is a three-dimensional
space of constant curvature. We suppose that $\MM$ is endowed with a
sufficiently smooth metric of the Lorentzian signature. It follows that if
space-time is homogeneous and isotropic, then a coordinate system $t,x^\mu$,
$\mu=1,2,3$ exists in the neighborhood of each point such that the metric takes
the form
\begin{equation}                                                  \label{qkdtwu}
  ds^2=dt^2+a^2\overset\circ g_{\mu\nu}dx^\mu dx^\nu,
\end{equation}
where $a(t)>0$ is an arbitrary function of time (the scale factor), and
$\overset\circ g_{\mu\nu}(x)$ is a negative-definite constant-curvature metric
on $\MS$ depending on the spatial coordinates $x\in\MS$ only.
\end{theorem}
Thus, the most general metric of a homogeneous and isotropic universe has form
(\ref{qkdtwu}) modulo a coordinate transformation.

This theorem is independent of the dimension of the manifold $\MM$ and the
signature of the metric $g$. The first condition of the theorem can be replaced
by the following: ``Let any constant-time slice $t\in\MR$ of the space-time
$\MM$ be space of constant curvature''. An exact definition of a homogeneous and
isotropic universe is given in Section \ref{sxzlpf}.

Theorem \ref{tbchyt} is fundamental in relativistic cosmology and is therefore
very important. Metric (\ref{qkdtwu}) was originally considered in [1--11].
We make a few comments on those parts of the original papers that are related to
the form of the metric.
\nocite{Friedm22R,Friedm24,Lemait27,Lemait33,Robert29,Robert33,Robert35}
\nocite{Tolman30A,Tolman30B,Tolman30C,Walker36}

Friedmann pioneered the use of metric (\ref{qkdtwu}) in cosmological models of
general relativity \cite{Friedm22R,Friedm24}. He did not write about homogeneous
and isotropic universe, instead, he simply required that all spatial cross
sections of constant time $t=\const$ are to be constant-curvature spaces and
assumed that the metric has form (\ref{qkdtwu}). In his first and second papers,
Friedmann considered spatial slices of positive and negative curvatures,
respectively.

Lemaitre analyzed solutions of the Einstein equations that describe closed
universe \cite{Lemait27}. However, he did not formulate theorem \ref{tbchyt}.
A more general class of cosmological models was considered in \cite{Lemait33},
but nevertheless the theorem was not yet formulated.

Robertson gave the theorem in both papers \cite{Robert29,Robert33}, but did not
provide a proof. Instead, he cited papers \cite{Hilber24,Fubini04}. The proof
consists of two parts. The first (Theorem \ref{tnbdht} in Section \ref{seqway})
was proved in general by Hilbert \cite{Hilber24}. The second part (Theorem
\ref{tsrewh} in Section \ref{seqway}) was proved in one direction by Fubini
\cite{Fubini04} (see also \cite{Eisenh26R}, Chapter VI, Exercise 3). Namely,
Fubini showed that metric (\ref{qkdtwu}) is homogeneous and isotropic, but the
converse proposition that {\em any} homogeneous and isotropic metric has this
form was not proved. In \cite{Robert35}, metric (\ref{qkdtwu}) was derived in a
different way by considering observers possessing certain properties. By
construction the resulting metric was homogeneous and isotropic. However,
Robertson proposed (see the discussion above Eqn (2.1) in \cite{Robert35}) that
the spatial part of the metric describes a constant-curvature space where the
curvature can take only discrete values $\pm1,0$, and therefore the general form
of metric (\ref{qkdtwu}) was not shown. Metric (\ref{emfrep}) below fits the
construction but in not of form (\ref{qkdtwu}).

Tolman obtained linear element (\ref{qkdtwu}) on a different basis
\cite{Tolman30A,Tolman30B,Tolman30C}. In particular, he assumed that there is a
spherical symmetry and that the time-coordinate lines are geodesics. In
addition, he required that the Einstein equations be satisfied. Homogeneity and
isotropy were not discussed in his papers.

In \cite{Walker36}, Section 10, Walker proved theorem \ref{tbchyt} in one
direction: metric (\ref{qkdtwu}) is homogeneous and isotropic. However, he did
not prove that {\em any} homogeneous and isotropic metric is of this form. In
fact, metric (\ref{emfrep}) given in Section 6 below satisfies Eqn (52) in
\cite{Walker36}, but is not of form (\ref{qkdtwu}).

Furthermore, I browsed more than 30 monographs on the general theory of
relativity, including my favorite books [15--19],
\nocite{LanLif62,Weinbe72,HawEll73,MiThWh73,Wald84} and found a proof of
theorem \ref{tbchyt} only in book \cite{Weinbe72}.

The aim of this paper is to explicitly review the general properties of the
Killing vector fields and necessary features of constant-curvature Riemannian
(pseudo-Riemannian) spaces. The proof of theorem \ref{tbchyt} is given by two
theorems \ref{tnbdht} and \ref{tsrewh}. The main idea behind the proof is
borrowed from \cite{Weinbe72}, but details are different. In particular, the
proof of theorem \ref{tnbdht} given in \cite{Eisenh26R} is simpler. In Section
\ref{sjhdtr}, we describe an example metric all of whose spatial slices are
constant-curvature spaces, but nevertheless the total metric is not
homogeneous and isotropic. Also, a new equivalent definition of a homogeneous
and isotropic metric is given in this section.

We hope that the reader is familiar with the basic concepts of differential
geometry, which can be found, i.e., in \cite{KobNom6369R,DuNoFo98R}.
\section{Killing vector fields                                   \label{skilve}}
We consider $n$-dimensional Riemannian (pseudo-Riemannian) manifold $(\MM,g)$
endowed with a metric $g(x)=g_{\al\bt}(x)dx^\al\otimes dx^\bt$,
$\al,\bt=0,1,\dotsc,n-1$ and the corresponding Levi--Civita connection $\Gamma$.
\begin{defn}
Diffeomorphism
\begin{equation*}
  \imath:\quad \MM\ni\quad x\mapsto x'=\imath(x)\quad\in\MM
\end{equation*}
is called an {\em isometry} of a Riemannian (pseudo-Riemannian) manifold
$(\MM,g)$ if the metric remains invariant
\begin{equation}                                                  \label{eisoma}
  g(x)=\imath^*g(x'),
\end{equation}
where $\imath^*$ is the induced map of differential forms.
\qed\end{defn}

Because an isometry leaves the metric invariant, it follows that all other
structures expressed in terms of the metric, like the Levi--Civita connection,
the geodesics, and the curvature tensor, are also invariant.

Map (\ref{eisoma}) can be represented in coordinate form. Let points $x^\al$ and
$x^{\prime\al}$ lie in the same coordinate neighborhood, and have coordinates
$x^\al$ and $x^{\prime \al}$. Then the isometry $\imath$ of the form
\begin{equation}                                                  \label{eisomt}
  g_{\al\bt}(x)=\frac{\pl x^{\prime\g}}{\pl x^\al}
  \frac{\pl x^{\prime\dl}}{\pl x^\bt}g_{\g\dl}(x'),
\end{equation}
relates the metric components in different points of the manifold.
\begin{prop}
All isometries of a given Riemannian (pseudo-Riemannian) manifold $(\MM,g)$
form an isometry group denoted by $\MI(\MM)\ni\imath$.
\end{prop}
\begin{proof}
The composition of two isometries is an isometry. The product of isometries is
associative. The identical map of $\MM$ is an isometry identified with the group
unit. Every isometry has an iverse, which is also an isometry.
\end{proof}

For a given metric, Eqn (\ref{eisomt}) defines functions $x'(x)$ that give an
isometry. In general, this equation has no solutions and the corresponding
manifold has no nontrivial isometries. In this case, the unit is a single
element of the isometry group. The larger the isometry group is, the smaller is
the class Riemannian (pseudo-Riemannian) manifolds.
\begin{exa}{\rm
Euclidean space $\MR^n$ endowed with the Euclidean metric $\dl_{\al\bt}$ has an
isometry group that is the inhomogeneous rotational group $\MI\MO(n,\MR)$,
$\dim\MI\MO(n,\MR)=\frac12n(n+1)$, consisting of rotations, translations, and
reflections.
}\qed\end{exa}

The isometry group $\MI(\MM)$ can be either a discrete group or a Lie group.
\begin{defn}
If the isometry group $\MI(\MM)$ is a Lie group, we can consider infinitesimal
transformations. In this case, we are dealing with infinitesimal isometries:
\begin{equation}                                                  \label{qbdvct}
  x^\al\mapsto x^{\prime\al}=x^\al+\e K^\al,\qquad\e\ll1.
\end{equation}
Any infinitesimal isometry is generated by a sufficiently smooth vector field
$K(x)=K^\al(x)\pl_\al$, which is called a {\em Killing vector field}.
\qed\end{defn}

Let $K=K^\al\pl_\al$ be a Killing vector field. Then invariance condition
(\ref{qbdvct}) takes the infinitesimal form
\begin{equation}                                                  \label{elidmk}
  \Lie_K g=0,
\end{equation}
where $\Lie_K$ is the Lie derivative along the vector field $K$. The coordinate
form is given by
\begin{equation}                                                  \label{ekileq}
   \nb_\al K_\bt+ \nb_\bt K_\al=0,
\end{equation}
where $K_\al:=K^\bt g_{\bt\al}$ is a Killing $1$-form, and the covariant
derivative
\begin{equation*}
  \nb_\al K_\bt:=\pl_\al K_\bt-\Gamma_{\al\bt}{}^\g K_\g
\end{equation*}
is defined by the Christoffel symbols $\Gamma_{\al\bt}{}^\g$ (the Levi--Civita
connection).
\begin{defn}
Equation (\ref{ekileq}) is called the {\em Killing equation} and integral curves
of a Killing vector field $K=K^\al\pl_\al$ are called {\em Killing
trajectories}. Any Killing vector field is in one-to-one correspondence with a
$1$-form $K=dx^\al K_\al$, where $K_\al:=K^\bt g_{\bt\al}$, which is called a
{\em Killing form}.
\qed\end{defn}

For any Riemannian (pseudo-Riemannian) manifold $(\MM,g)$, Killing equation
(\ref{elidmk}) always has the trivial solution $K=0$. If the equation has zero
solution only, there is no nontrivial continious symmetries.

The Killing trajectories $\lbrace x^\al(\tau)\rbrace\in\MM$ with $\tau\in\MR$
are defined by the system of ordinary differential equations
\begin{equation}                                                  \label{qlkodn}
  \dot x^\al=K^\al,
\end{equation}
which has a unique solution passing through a point
$p=\lbrace p^\al\rbrace\in\MM$ for any differentiable Killing vector field. For
small $\tau\ll1$, the trajectory has the form
\begin{equation}                                                  \label{ekiltr}
  x^\al(t)=p^\al+\tau K^\al(p)+\osmall(\tau),
\end{equation}
where the integration constant is chosen such that the trajectory goes through
the point $p$ at $\tau=0$.

Any Killing vector field generates a one-parameter subgroup of the isometry
group. If a Killing vector field vanishes at some point, this point is a
stationary point under the action of the isometry group generated by the vector
field. Killing vector fields are called complete if Killing trajectories are
defined for all $\tau\in\MR$. They must have this property because the
isometries form a group.

A given Killing vector field defines not only an infinitesimal symmetry but
also the whole one-parameter subgroup of $\MI(\MM)$. For this, we have to find
integral curves (Killing trajectories) passing through any point $p\in\MM$. If
$x(0)=p$, then there is a diffeomorphism
\begin{equation*}
  \imath:\quad \MM\ni\quad p\mapsto x(t)\quad\in\MM,
\end{equation*}
corresponding to each value $\tau\in\MR$.

The contravariant component form of Killing equation (\ref{ekileq}) is given by
\begin{equation}                                                  \label{ekilco}
  g_{\al\g}\pl_\bt K^\g+g_{\bt\g}\pl_\al K^\g+K^\g\pl_\g g_{\al\bt}=0.
\end{equation}
This equation is linear in both the Killing vector and the metric. It follows
that any two metrics differing by a prefactor have the same Killing
trajectories. Moreover, Killing vector fields are defined modulo an arbitrary
nonzero constant prefactor. In particular, if $K$ is a Killing vector, then $-K$
is also a Killing vector. If there are several Killing vector fields, their
linear combination is again a Killing vector field. In other words, Killing
fields form a vector space over real numbers, which is a subspace of the vector
space of all vector fields $\CX(\MM)$ on the manifold $\MM$. This space is
endowed with a bilinear form. It is easy to show that the commutator of two
Killing vector fields $K_1$ and $K_2$ is another Killing vector field:
\begin{equation*}
  \Lie_{[K_1,K_2]}g=\Lie_{K_1}\circ\Lie_{K_2}g-\Lie_{K_2}\circ\Lie_{K_1}g=0.
\end{equation*}
It follows that Killing vector fields form a Lie algebra $\Gi(\MM)$ over real
numbers, which is a subalgebra in the infinite-dimensional Lie algebra of all
vector fields, $\Gi(\MM)\subset\CX(\MM)$. This is the Lie algebra of the
isometry Lie group $\MI(\MM)$.

\begin{prop}
Let a Riemannian (pseudo-Riemannian) manifold $(\MM,g)$ have $\Sn\le\dim\MM$
nonvanishing commuting  and linearly independent Killing vector fields $K_i$,
$i=1,\dotsc,\Sn$. A coordinate system then exists such that the metric is
independent of $\Sn$ coordinates corresponding to Killing trajectories. The
converse statement is that if the metric is independent of $\Sn$ coordinates in
some coordinate system, then the metric $g$ has at least $\Sn$ nonvanishing
commuting Killing vector fields.
\end{prop}
\begin{proof}
For any nonvanishing Killing vector field, there is a coordinate system where
the field has components $(1,0,\dotsc,0)$. For a set of independent commuting
vector fields $K_i$ this implies that there is a coordinate system
$x^1,\dotsc,x^n$ such that each Killing vector field has just one nontrivial
component $K_i=\pl_i$. In this coordinate system Killing equation (\ref{ekilco})
is particularly simple:
\begin{equation}                                                  \label{qvdgfr}
  \pl_i g_{\al\bt}=0,\qquad i=1,\dotsc,\Sn\le\dim\MM.
\end{equation}
This implies that the metric components are independent of the coordinates
$x^i$.

In this coordinate system, the Killing trajectories are given by the equations
\begin{equation*}
  \dot x^i=1,\quad\dot x^\mu=0,\qquad\mu\ne i.
\end{equation*}
We see that the coordinate lines $x^i$ are Killing trajectories.

As regards the converse statement, if the metric components are independent of
$\Sn$ coordinates, then Eqns (\ref{qvdgfr}) are satisfied. These are the
Killing equations for commuting vector fields $K_i:=\pl_i$.
\end{proof}

It follows that in the limit case where the number of commuting Killing vectors
is equal to the dimension of the manifold $\Sn=n$, a coordinate system exists
such that the corresponding metric components are constant.
\begin{exa}{\rm
In the Euclidean space $\MR^n$, the metric in Cartesian coordinates $x^\al$,
$\al=1,\dotsc,n$ has constant components $g_{\al\bt}=\dl_{\al\bt}$. This metric
has $n$ commuting Killing vectors $K_\al:=\pl_\al$, which generate translations.
All coordinate lines are Killing trajectories.}
\qed\end{exa}

If a Riemannian manifold $(\MM,g)$ has two or more noncommuting Killing vectors,
this does not mean that there is a coordinate system such that the metric
components are independent of two or more coordinates.
\begin{exa}{\rm
We consider the two-dimensional sphere $\MS^2\hookrightarrow\MR^3$ embedded into
a three-dimensional Euclidean space in the standard manner. Let the metric $g$
on the sphere be the induced metric. then the Riemannian manifold $(\MS^2,g)$
has three noncommuting Killing vector fields corresponding to the rotation group
$\MS\MO(3)$. It is obvious that there is no such coordinate system where the
metric components are independent of two coordinates. Indeed, in such a
coordinate system, the metric components are constant and therefore the
curvature is zero. But this is impossible because the curvature of the sphere is
nonzero.}
\qed\end{exa}

In general relativity, we suppose that the space-time is a pseudo-Riemannian
manifold $(\MM,g)$ endowed with a metric of the Lorentzian signature. Using the
notion of a Killing vector field, we can give the following invariant
definition.
\begin{defn}
A space-time $(\MM,g)$ is called {\em stationary} if there is a time-like
Killing vector field.
\qed\end{defn}

Killing vector fields have a number of remarkable features. We consider the
simplest of them.
\begin{prop}                                                      \label{poriki}
The length of a Killing vector along the Killing trajectory is constant:
\begin{equation}                                                  \label{ekilen}
  \Lie_K K^2=\nb_K K^2=K^\al\pl_\al K^2=0.
\end{equation}
\end{prop}
\begin{proof}
Contracting Eqns (\ref{ekileq}) with $K^\al K^\bt$ yields the equalities
\begin{equation*}                                                    \tag*{\qed}
  2K^\al K^\bt\nb_\al K_\bt=K^\al\nb_\al K^2=K^\al\pl_\al K^2=0.
\end{equation*}
\renewcommand{\qed}{}\end{proof}
\begin{cor}
Killing vector fields on a Lorentzian manifold are oriented time-like,
light-like, or space-like.
\end{cor}

A metric on a manifold defines two particular types of curves: geodesics
(extremals) and Killing trajectories, if they exist.
\begin{prop}                                                      \label{pkiext}
Let $(\MM,g)$ be a Riemannian (pseudo-Riemannian) manifold with a Killing vector
field $K$. A Killing trajectory is geodesic if and only if the length of the
Killing vector is constant on $\MM$: $K^2=\const$ for all $x\in\MM$.
\end{prop}
\begin{proof}
A Killing trajectory $x^\al(\tau)$ is given by the system of equations
\begin{equation}                                                  \label{eqkilt}
  \dot x^\al=K^\al.
\end{equation}
The length of the infinitesimal interval of the Killing trajectory
\begin{equation*}
  ds^2=g_{\al\bt}\dot x^\al\dot x^\bt d\tau^2=K^2d\tau^2
\end{equation*}
is constant along the trajectory, i.e., the parameter $\tau$ is proportional to
the length of the trajectory and is therefore canonical. Differentiating
equation (\ref{eqkilt}) with respect to canonical parameter $\tau$ yields the
relation
\begin{equation*}
  \ddot x^\al=\pl_\bt K^\al\dot x^\bt
  =(\nb_\bt K^\al-\Gamma_{\bt\g}{}^\al K^\g)\dot x^\bt,
\end{equation*}
which can be represented as
\begin{equation}                                                  \label{eqslik}
  \ddot x^\al=K^\bt\nb_\bt K^\al-\Gamma_{\bt\g}{}^\al
  \dot x^\bt\dot x^\g.
\end{equation}
Using Killing equatins, we can rewrite the first term in the right-hand side as
\begin{equation*}
  K^\bt\nb_\bt K^\al=-\frac12g^{\al\bt}\pl_\bt K^2.
\end{equation*}
Then Eqn (\ref{eqslik}) takes the form
\begin{equation*}
  \ddot x^\al=-\frac12g^{\al\bt}\pl_\bt K^2-\Gamma_{\bt\g}{}^\al
  \dot x^\bt\dot x^\g.
\end{equation*}
The last equation coincides with the geodesic equation if and only if
$K^2=\const$.
\end{proof}
It follows that Killing trajectories differ from the geodesics in general.
\begin{exa}{\rm
We consider the Euclidean plane $\MR^2$ endowed with Euclidean metric. This
metric is invariant under three-parameter inhomogeneous rotation group
$\MI\MO(2)$. We let $x,y$ and $r,\vf$ denote Cartesian and polar coordinates on
the plane. Then the Killing vector fields corresponding to rotations and
translations are $K_1=\pl_\vf$ and $K_2=\pl_x$, $K_3=\pl_y$. The squared vector
norms are
\begin{equation*}
  K^2_1=r^2,\qquad K^2_2=K^2_3=1.
\end{equation*}
The Killing vector fields $K_2$ and $K_3$ have a constant length on the whole
plane. Their trajectories are straight lines, which are geodesics. This agrees
with proposition \ref{pkiext}. The Killing trajectories corresponding to
rotations $K_1$ are concentric circles around the origin. In accordance with
proposition \ref{poriki}, the length of the Killing vector $K_1$ is constant
along the circles, but nonconstant on the whole plane $\MR^2$. The corresponding
Killing trajectories are circles, which are not geodesics.}
\qed\end{exa}
\begin{exa}{\rm
We consider a semisimple Lie group $\MG$ as a Riemannian (pseudo-Rieman\-nian)
manifold endowed with the Cartan--Killing form as an invariant metric. Then the
left-invariant and right-invariant vector fields on $\MG$ generate right and
left group actions. Both left and right group actions leave the metric
invariant. Therefore, left- and right-invariant vector fields are Killing vector
fields. Their length equals $\pm1$. Hence, the corresponding Killing
trajectories are geodesics.}
\qed\end{exa}

Contracting Killing equation (\ref{ekileq}) with the metric shows that the
divergence of a Killing vector field is zero:
\begin{equation}                                                  \label{edikiv}
  \nb_\al K^\al=0.
\end{equation}
The covariant derivative $\nb^\bt$ of Killing equation (\ref{ekileq}) takes the
form
\begin{equation*}
  \nb^\bt(\nb_\bt K_\al+\nb_\al K_\bt)=\triangle K_\al
  +(\nb^\bt\nb_\al-\nb_\al\nb^\bt)K_\bt=0,
\end{equation*}
where we used relation (\ref{edikiv}), and where $\triangle:=\nb^\bt\nb_\bt$ is
the Laplace--Beltrami operator on the manifold $\MM$. Using equality
\begin{equation*}
  [\nb_\al,\nb_\bt]K_\g=-R_{\al\bt\g}{}^\dl K_\dl
\end{equation*}
for the commutator of covariant derivatives, we arrive at the following equation
for the Killing vector:
\begin{equation}                                                  \label{eivkil}
  \triangle K_\al=R_{\al\bt} K^\bt,
\end{equation}
where $R_{\al\bt}:=R_{\al\g\bt}{}^\g$ is the Ricci tensor.

In the case of a constant-curvature space, the Ricci tensor is proportional to
the scalar curvature [see Eqn (\ref{ericsk}) below], and therefore Eqn
(\ref{eivkil}) is simplified to
\begin{equation*}
  \triangle K_\al=\frac Rn K_\al,\qquad R=\const.
\end{equation*}
In other words, each component of the Killing vector is an eigenfunction of the
Laplace--Beltrami operator.
\begin{prop}
Let $X,Y\in\CX(\MM)$ be two arbitrary vector fields on a Riemannian
(pseudo-Riemannian) manifold $(\MM,g)$ and $K$ be a Killing vector. Then the
following equality holds:
\begin{equation*}
  g\big((\Lie_K-\nb_K)X,Y\big)+g\big(X,(\Lie_K-\nb_K)Y\big)=0,
\end{equation*}
where $\Lie_KX=[K,X]$ is the Lie derivative and
$\nb_K X=K^\al(\pl_\al X^\bt+\Gamma_{\al\g}{}^\bt X^\g)\pl_\bt$ is the covariant
derivative of a vector field $X$ along the Killing vector field $K$.
\end{prop}
\begin{proof}
Direct verification using the Christoffel symbols and Killing equation
(\ref{ekileq}).
\end{proof}
\section{Homogeneous and isotropic spaces                        \label{sxzlpf}}
Killing equation (\ref{ekileq}) imposes severe restrictions on Killing vector
fields, which we have to discuss. Using the formula for the commutator of
covariant derivatives, we find the relation
\begin{equation}                                                  \label{eqcoki}
  \nb_\al \nb_\bt K_\g-\nb_\bt \nb_\al K_\g=-R_{\al\bt\g}{}^\dl K_\dl.
\end{equation}
Then, applying the identity
\begin{equation*}
  R_{\al\bt\g}{}^\dl+R_{\bt\g\al}{}^\dl+R_{\g\al\bt}{}^\dl=0
\end{equation*}
for the curvature tensor and Killing equation (\ref{ekileq}), we find the
equality
\begin{equation*}
  \nb_\al\nb_\bt K_\g+\nb_\bt\nb_\g K_\al+\nb_\g\nb_\al K_\bt=0,
\end{equation*}
where the terms are related by cyclic permutations. Using this equality, we can
represent equation (\ref{eqcoki}) as
\begin{equation}                                                  \label{emaeqk}
  \nb_\g\nb_\al K_\bt= R_{\al\bt\g}{}^\dl K_\dl.
\end{equation}
Contracting the indices $\g$ and $\al$, we obtain exactly equality
(\ref{eivkil}) from Section (\ref{skilve}).

Equation (\ref{emaeqk}) follows from the Killing equations. However, they are
not equivalent. Nevertheless, Eqn (\ref{emaeqk}) has important consequences.
We assume that Killing vector fields are real analytical functions< i.e, their
components can be represented as Taylor series converging in some neighborhood
$\MU_p$ of a point $p\in\MM$. We suppose that all components of the Killing
$1$-form $K_\al(p)$ and their first derivatives $\pl_\bt K_\al(p)$ are given
at some fixed point $p\in\MM$. Then the second partial derivatives of the
Killing $1$-form $\pl^2_{\bt\g}K_\al$ can be found from Eqn (\ref{emaeqk}). Now,
we evaluate the covariant derivative of Eqn (\ref{emaeqk}), thereby obtaining
some relation for the third derivatives and so on, up to infinity. It is
important that all the relations are linear in the Killing vector components and
their derivatives. It follows that the Killing $1$-form components in some
neighborhood $\MU_p$ are of the form
\begin{equation}                                                  \label{ekilde}
  K_\al(x,p)=A_\al{}^\bt(x,p)K_\bt(p)
  +B_{\al}{}^{\bt\g}(x,p)\big[\pl_\bt K_\g(p)-\pl_\g K_\bt(p)\big],
\end{equation}
where $A_\al{}^\bt(x,p)$ and $B_{\al}{}^{\bt\g}(x,p)$ are some functions. The
antisymmetry in the indices $\bt$ and $\g$ in the last term is achieved by
expressing the symmetrized partial derivative in terms of the Killing vector
components by means of Killing equations (\ref{ekileq}). Therefore, Killing
$1$-form components in some neighborhood $\MU_p$ are linear combinations of the
Killing form and their exterior derivative at the point $p$.

The Killing form $K_\al(x,p)$ depends on two arguments. The second variable
$p$ shows that the form has properties specified at the point $p\in\MM$. By
assumption, representation (\ref{ekilde}) holds at any point $p\in\MM$: it is
just necessary to know the values $K(p)$ and $dK(p)$. We suppose that the
functions $K_\al(x,p)$ are real analytic in both variables $x$ and $p$.

It is assumed that the Killing form components can be expanded into Taylor
series near any point $p\in\MM$. Let $\MU_p$ be a neighborhood of a point $p$
where representation (\ref{ekilde}) holds and is invertible, i.e., the variables
$x$ and $p$ can be replaced with some new functions $A$ and $B$. We consider a
point $q$ outside $\MU_p$. For this point, the invertible representation like
(\ref{ekilde}) is also true in some neighborhood $\MU_q$. We suppose that the
point $q$ lies close enough to $\MU_p$ such that the neighborhoods overlap,
$\MU_p\cap\MU_q\ne\emptyset$. Then, for any point belonging to the intersection
$x\in\MU_p\cap\MU_q$, representation (\ref{ekilde}) holds with respect to
components of $K(p)$ and $K(q)$ and their exterior derivatives. We see that the
Killing form and its exterior derivative at $q$ can be linearly expressed in
terms of their values at $p$. Therefore, representation (\ref{ekilde}) holds in
the union $\MU_p\cup\MU_q$. This construction can be extended to the whole
manifold $\MM$. As a result, representation (\ref{ekilde}) holds for all points
$x,p\in\MM$.

We now assume that a Riemannian (pseudo-Riemannian) manifold $(\MM,g)$ has
several Killing vector fields $K_i$, $i=1,\dotsc,\Sn$. Then representation
(\ref{ekilde}) holds for each Killing vector:
\begin{equation}                                                  \label{ekildr}
  K_{i\al}(x,p)=A_\al{}^\bt(x,p)K_{i\bt}(p)
  +B_{\al}{}^{\bt\g}(x,p)\big[\pl_\bt K_{i\g}(p)-\pl_\g K_{i\bt}(p)\big].
\end{equation}
The functions $A_\al{}^\bt(x,p)$ and $B_{\al}{}^{\bt\g}(x,p)$ are the same for
any Killing form, because they are defined by relation (\ref{emaeqk}), which is
linear in the Killing form components and their derivatives. They are uniquely
defined by the metric, the curvature, and its covariant derivatives. It is
supposed that in the resulting representation, the point $p\in\MM$ is arbitrary
but fixed, while the point $x\in\MM$ ranges the whole manifold $\MM$.

Equality (\ref{emaeqk}) is a system of partial differential equations for the
Killing form components and has nontrivial integrability conditions. One of them
has the covariant form
\begin{equation*}
  \left[\nb_\g,\nb_\dl\right]\nb_\al K_\bt
  =-R_{\g\dl\al}{}^\e\nb_\e K_\bt-R_{\g\dl\bt}{}^\e\nb_\al K_\e,
\end{equation*}
where the square brackets denote the commutator of covariant derivatives.
Substituting the initial equation (\ref{emaeqk}) for the second derivatives in
the left-hand side of this equation, by straightforward computation we find that
\begin{equation}                                                  \label{erelkd}
  \left(R_{\al\bt\g}{}^\e\dl_\dl^\z-R_{\al\bt\dl}{}^\e\dl_\g^\z
  +R_{\g\dl\al}{}^\e\dl_\bt^\z-R_{\g\dl\bt}{}^\e\dl_\al^\z\right)\nb_\z K_\e
  =\left(\nb_\g R_{\al\bt\dl}{}^\e-\nb_\dl R_{\al\bt\g}{}^\e\right)K_\e.
\end{equation}
When the curvature is nontrivial, this equation is a linear relation between
components of the Killing form $K_\al$ and their covariant derivatives
$\nb_\bt K_\al$. Conversely, if we know some properties of the Killing form,
the resulting equality can define the structure of the curvature tensor. In
Theorem \ref{tkildi} in what follows, Eqn (\ref{erelkd}) is used to prove the
statement that homogeneous and isotropic manifold is a constant-curvature space.
\begin{defn}
A Riemannian (pseudo-Riemannian) manifold $(\MM,g)$ of dimension $\dim\MM=n$ is
called {\em homogeneous at a point} $p\in\MM$ if there are infinitesimal
isometries mapping this point to any other point in some neighborhood $\MU_p$ of
$p$. In other words, the metric should have Killing vector fields with arbitrary
direction at $p$. Because Killing vectors form a linear space, it is necessary
and sufficient to have a set on $n$ Killing forms in the dual space
$K^{(\g)}=dx^\al K_\al{}^{(\g)}(x,p)$, where the index $\g$ in parenthesis
labels Killing forms, such that the following relations are satisfied:
\begin{equation}                                                  \label{ehocok}
  K_\al{}^{(\g)}(p,p)=\dl_\al^\g.
\end{equation}
If a Riemannian (pseudo-Riemannian) space $(\MM,g)$ is homogeneous at any point
$x\in\MM$, it is called {\em homogeneous}. In other words, the isometry group
act on $\MM$ transitively.

A Riemannian (pseudo-Riemannian) manifold $(\MM,g)$ is called {\em isotropic at
a point} $p\in\MM$ if there are infinitesimal isometries with Killing forms
$K(x,p)$ such that the given point is stable, i.e., $K(p,p)=0$, and the exterior
derivative $dK(x,p)$ at $p$ takes all possible values in the space of $2$-forms
$\Lm_2(\MM)|_p$ at $p$. This happens if and only if there is a set of
$\frac12n(n-1)$ Killing forms
$K^{[\g\dl]}=-K^{[\dl\g]}=dx^\al K_\al{}^{[\g\dl]}(x,p)$, where indices $\g,\dl$
label Killing forms, such that the following relations are satisfied:
\begin{equation}                                                  \label{eprkif}
\begin{split}
  K_\al{}^{[\g\dl]}(p,p)&=0,
\\
  \left.\frac{\pl K_\bt{}^{[\g\dl]}(x,p)}{\pl x^\al}\right|_{x=p}
  &=\dl_{\al\bt}^{\g\dl}-\dl_{\al\bt}^{\dl\g}.
\end{split}
\end{equation}
If a Riemannian (pseudo-Riemannian) manifold $(\MM,g)$ is isotropic at any
point, it is called {\em isotropic}.
\qed\end{defn}

By continuity, it follows that the forms $K^{(\g)}$ and $K^{[\g\dl]}$ are
linearly independent in some neighborhood of point $p$.
\begin{prop}                                                      \label{pisoho}
Any isotropic Riemannian (pseudo-Riemannian) manifold $(\MM,g)$ is also
homogeneous.
\end{prop}
\begin{proof}
If a manifold is isotropic, the Killing forms $K^{[\g,\dl]}(x,p)$ and
$K^{[\g,\dl]}(x,p+dp)$ satisfy Eqns (\ref{eprkif}) in some neighborhoods of the
respective points $p$ and $p+dp$. Their arbitrary linear combination and
therefore arbitrary linear combination of derivatives
\begin{equation*}
  c^\al\frac{\pl K_\bt{}^{[\g\dl]}(x,p)}{\pl p^\al}
  :=c^\al\underset{dp^\al\to 0}\lim\frac
  {K_\bt{}^{[\g,\dl]}(x,p+dp)-K_\bt{}^{[\g,\dl]}(x,p)}{dp^\al}
\end{equation*}
are Killing forms for arbitrary constants $c^\al$. We differentiate the Killing
form $K^{[\g\dl]}$ with respect to $x$ at the point $p$. From the first relation
in (\ref{eprkif}), it follows that
\begin{equation*}
  \frac\pl{\pl p^\al}K_\bt{}^{[\g\dl]}(p,p)
  =\left.\frac{\pl K_\bt{}^{[\g\dl]}(x,p)}{\pl x^\al}\right|_{x=p}
  +\left.\frac{\pl K_\bt{}^{[\g\dl]}(x,p)}{\pl p^\al}\right|_{x=p}=0.
\end{equation*}
Using the second condition in (\ref{eprkif}), we obtain the equality
\begin{equation*}
  \left.\frac{\pl K_\bt{}^{[\g\dl]}(x,p)}{\pl p^\al}\right|_{x=p}=
  -\dl_{\al\bt}^{\g\dl}+\dl_{\al\bt}^{\dl\g}.
\end{equation*}
Now, from $K^{[\g\dl]}$ we can build Killing forms that take arbitrary values
$dx^\al a_\al$ at point $p$, where $a_\al\in\MR$. For this, it is sufficient to
assume that
\begin{equation*}
  K_\al:=\frac{a_\g}{n-1}\frac{\pl K_\al{}^{[\g\dl]}(x,p)}{\pl p^\dl}.
\end{equation*}
By choosing appropriate constants $a_\g$, we find a set of Killing forms
satisfying equalities (\ref{ehocok}).
\end{proof}
Thanks to this theorem, it suffices to use the term ``isotropic universe''.
However, we prefer to call it ``homogeneous and isotropic'', because this name
emphasizes important physical properties.
\begin{theorem}                                                   \label{tkildi}
The Lie algebra of infinitesimal isometries $\Gi(\MM)$ of a connected Riemannian
(pseudo-Riemannian) manifold $\MM$ has the dimension not exceeding
$\frac12n(n+1)$, where $n:=\dim\MM$. If the dimensional is maximal,
$\dim\Gi(\MM)=\frac12n(n+1)$, then the manifold $\MM$ is homogeneous and
isotropic, being a constant-curvature space.
\end{theorem}
\begin{proof}
The dimension of the Lie algebra $\Gi(\MM)$ is equal to the maximal number of
linearly independent Killing vector fields on the manifold $\MM$. From Eqn
(\ref{ekildr}), it follows that the number $\Sn$ of linearly independent Killing
vectors cannot exceed the number of independent components of Killing forms
$\lbrace K_\al(p)\rbrace$ and their exterior derivatives
$\lbrace \pl_\bt K_\al(p)-\pl_\al K_\bt(p)\rbrace$ at a fixed point $p\in\MM$.
The number of independent components of an arbitrary 1-form does not exceed $n$,
and the number of independent components of its exterior derivative cannot
exceed $\frac12n(n-1)$. Thus, we find a restriction of the dimension of the Lie
algebra of isometries generated by Killing vector fields:
\begin{equation*}
  \dim\Gi(\MM)\le n+\frac12n(n-1)=\frac12n(n+1).
\end{equation*}
This proves the first statement of the theorem. The real analyticity is
important here, because it was used to obtain representation (\ref{ekildr}).

The connectedness of the manifold $\MM$ guarantees that a number of Killing
vector fields is defined unambiguously. If $\MM$ has several connected
components, the number of independent Killing vector fields may depend on a
particular component.

There are at most $\frac12n(n+1)$ independent Killing vector fields on
homogeneous and isotropic manifolds. By Eqn (\ref{ekildr}) they define all
possible Killing vector fields on the manifold $\MM$. Consequently, if the
manifold has the maximal number of independent Killing fields, it is necessarily
homogeneous and isotropic.

We now prove that any homogeneous and isotropic manifold is a constant-curvature
space. If a manifold is homogeneous and isotropic, for any point $x\in\MM$
there are Killing forms such that $K_\al(x)=0$, while their derivatives
$\nb_\bt K_\al(x)$ can be arranges into antisymmetric matrix. As a consequence,
an antisymmetric coefficient at $\nb_\z K_\e$ in Eqn (\ref{erelkd}) must be
zero. It follows that
\begin{equation}                                                  \label{erecuk}
   R_{\al\bt\g}{}^\e\dl_\dl^\z-R_{\al\bt\dl}{}^\e\dl_\g^\z
  +R_{\g\dl\al}{}^\e\dl_\bt^\z-R_{\g\dl\bt}{}^\e\dl_\al^\z
  =R_{\al\bt\g}{}^\z\dl_\dl^\e-R_{\al\bt\dl}{}^\z\dl_\g^\e
  +R_{\g\dl\al}{}^\z\dl_\bt^\e- R_{\g\dl\bt}{}^\z\dl_\al^\e.
\end{equation}
If the space is homogeneous and isotropic, then for any point $x\in\MM$ there
are Killing forms taking arbitrary values at this point. From Eqns
(\ref{erelkd}) and (\ref{erecuk}), it follows that
\begin{equation}                                                  \label{erecud}
  \nb_\g R_{\al\bt\dl}{}^\e=\nb_\dl R_{\al\bt\g}{}^\e.
\end{equation}
In Eqn (\ref{erecuk}), we contract the indices $\dl$ and $\z$ and then lower the
upper index. As a result, we find the curvature tensor expressed in terms of the
Ricci tensor and the metric
\begin{equation}                                                  \label{ecurik}
  (n-1)R_{\al\bt\g\dl}=R_{\bt\dl}g_{\al\g}-R_{\al\dl}g_{\bt\g}.
\end{equation}
Because the right-hand side of formula (\ref{ecurik}) has to be antisymmetric in
$\dl$ and $\g$, there is an additional restriction
\begin{equation*}
  R_{\bt\dl}g_{\al\g}-R_{\al\dl}g_{\bt\g}=
  -R_{\bt\g}g_{\al\dl}+ R_{\al\g}g_{\bt\dl}.
\end{equation*}
Contracting the indices $\bt$ and $\g$ yields a relation between the Ricci
tensor and the scalar curvature
\begin{equation}                                                  \label{ericsk}
  R_{\al\dl}=\frac1n Rg_{\al\dl},
\end{equation}
where $R:=g^{\al\bt}R_{\al\bt}$ is the scalar curvature. Substituting the above
relation in equality (\ref{ecurik}) results in the following expression for the
full curvature tensor:
\begin{equation}                                                  \label{ecumek}
  R_{\al\bt\g\dl}=\frac R{n(n-1)}
  \left(g_{\al\g}g_{\bt\dl}-g_{\al\dl}g_{\bt\g}\right).
\end{equation}

Now, to complete the proof, we have to show that the scalar curvature $R$ of a
homogeneous and isotropic space is constant. For this, we use the contracted
Bianchi identity
\begin{equation*}
  2\nb_\bt R_\al{}^\bt-\nb_\al R=0.
\end{equation*}
Substituting formula (\ref{ericsk}) for the Ricci tensor in this identity yields
the equation
\begin{equation*}
  \left(\frac2n-1\right)\pl_\al R=0.
\end{equation*}
For $n\ge3$, it follows that $R=\const$.

The case $n=2$ is to be considered separately. Contracting the indices $\bt$ and
$\e$ in Eqn (\ref{ecurik}) yields the equality
\begin{equation*}
  \nb_\g R_{\al\dl}-\nb_\dl R_{\al\g}=0.
\end{equation*}
Then, contracting with $g^{\al\dl}$ and using relation (\ref{ericsk}) yields the
equation $\pl_\g R=0$, and hence, $R=\const$ also in the case $n=2$.

Thus, the scalar curvature in (\ref{ecumek}) has to be constant, $R=\const$, and
therefore a maximally symmetric Riemannian (pseudo-Riemannian) manifold is a
constant-curvature space.
\end{proof}

\begin{exa}{\rm
We consider the Euclidean space $\MR^n$ with a zero-curvature metric, i.e.,
$R_{\al\bt\g\dl}=0$. This space obviously has a constant curvature. It follows
that there is a coordinate system $x^\al$, $\al=1,\dotsc,n$ such that all metric
components are constant. The Christoffel symbols in this coordinate system are
zero. Equation (\ref{emaeqk}) for Killing vector fields takes the simple form
\begin{equation*}
  \pl^2_{\bt\g}K_\al=0.
\end{equation*}
The general solution of this equation is linear in coordinates:
\begin{equation*}
  K_\al(x)=a_\al+b_{\al\bt}x^\bt,
\end{equation*}
where $a_\al$ and $b_{\al\bt}$ are some constants. It follows from Killing
equation (\ref{ekileq}) that this expression defines the Killing form if and
only if the matrix $b_{\al\bt}$ is anisymmetric, i.e., $b_{\al\bt}=-b_{\bt\al}$.
Therefore, we can define $\frac12n(n+1)$ linearly independent Killing forms:
\begin{align*}
  K_\al{}^{(\g)}(x)&=\dl_\al^\g,
\\
  K_\al{}^{[\g\dl]}(x)&=\dl_\al^\dl x^\g-\dl_\al^\g x^\dl.
\end{align*}
Hence, an arbitrary Killing form is the linear combination
\begin{equation*}
  K_\al=a_\g K_\al{}^{(\g)}+\frac12b_{\dl\g}K_\al{}^{[\g\dl]}.
\end{equation*}
Here, the $n$ Killing vectors $K^{(\g)}$ generate translations along coordinate
axes in $\MR^n$, while the $\frac12n(n-1)$ Killing vectors $K^{[\g\dl]}$
generate rotations around the origin. Thus, a zero-curvature metric has the
maximal number $\frac12n(n+1)$ of Killing vectors, and therefore the space is
homogeneous and isotropic.

It is known that the metric can be diagonalized by linear coordinate
transformations, such that the main diagonal elements are $\pm1$, depending on
the metric signature. If the metric is Riemannian (positive definite), it can be
mapped into the standard form $g_{\al\bt}=\dl_{\al\bt}$. This metric is
invariant under the inhomogeneous rotation group $\MI\MO(n)$.}
\qed\end{exa}

We have proved that a homogeneous and isotropic space has constant curvature.
The converse is also true. This can be formulated in several steps.
\begin{theorem}
Let $(\MM,g)$ be a Riemannian (pseudo-Riemannian) space of constant curvature
with the scalar curvature tensor like (\ref{ecumek}), where $R=\const$ is the
scalar curvature. We assume that the metric signature is $(p,q)$. Then, in some
neighborhood of a point $x\in\MM$, there is a coordinate system (stereographic
coordinates) such that the metric is given by
\begin{equation}                                                  \label{qbcvdy}
  ds^2=\frac{\eta_{\al\bt}dx^\al dx^\bt}{\left(1-\frac{Rx^2}8\right)^2},
\end{equation}
where
\begin{equation*}
  \eta:=\diag(\underbrace{+\dotsc+}_p\underbrace{-\dotsc-}_q),\qquad
  x^2:=\eta_{\al\bt}x^\al x^\bt.
\end{equation*}
\end{theorem}
\begin{proof}
(See, i.e., Theorem 2.4.12 in \cite{Wolf72}).
\end{proof}

If $R=0$, then the full curvature tensor (\ref{ecumek}) is also zero. It
follows that a zero-curvature space is locally isomorphic to the Euclidean
(pseudo-Euclidean) space $\MR^{p,q}$, and formula (\ref{qbcvdy}) holds.

We consider the case $R\ne0$. Metric (\ref{qbcvdy}) is the induced metric on the
sphere $\MS^{p+q}$ or the hyperboloid $\MH^{p+q}$ embedded into the
larger-dimension pseudo-Euclidean space $\MR^{p+1,q}$. Indeed, let $u,x^\al$ be
Cartesian coordinates in $\MR^{p+1,q}$. The metric then takes the form
\begin{equation}                                                  \label{qxbvct}
  ds^2:=du^2+\eta_{\mu\nu}dx^\mu dx^\nu.
\end{equation}
We consider the sphere (hyperboloid) embedded into the Euclidean
(pseudo-Euclidean) space $\MR^{p+1,q}$ by means of the equation
\begin{equation}                                                  \label{qbdftr}
  u^2+\eta_{\mu\nu}x^\mu x^\nu=b,\qquad b=\const\ne0.
\end{equation}

To simplify the calculations, we ignore the signs and domains of the radicand,
which depend on the constant $b$ and the signature of the metric
$\eta_{\mu\nu}$. Both signs and signature can be properly dealt with in each
particular case.

We introduce spherical coordinates
$\lbrace x^\al\rbrace\mapsto\lbrace r,\chi^1,\dotsc,\chi^{p+q-1}\rbrace$, where
$r$ is the radial coordinate and $\chi$ denotes angular coordinates in the
Euclidean (pseudo-Euclidean) space $\MR^{p,q}\subset\MR^{p+1,q}$. Then metric
(\ref{qxbvct}) and embedding equation (\ref{qbdftr}) take the form
\begin{align}                                                     \label{qxncpo}
  &ds^2=du^2+dr^2+r^2d\Om,
\\                                                                \label{qbvcoe}
  &u^2+r^2=b,
\end{align}
where $d\Om(\chi,d\chi)$ is the angular part of the Euclidean metric (whose
explicit form is not important here). Equation (\ref{qbvcoe}) yields the
relations
\begin{equation*}
  u=\pm\sqrt{b-r^2}\qquad\Rightarrow\qquad
  du=\mp\frac{rdr}{\sqrt{b-r^2}}.
\end{equation*}
Substituting $du$ in (\ref{qxncpo}) yields the induced metric
\begin{equation}                                                  \label{qbcvft}
  ds^2=\frac{bdr^2}{b-r^2}+r^2d\Om.
\end{equation}
Now, we transform the radial coordinate $r\mapsto\rho$ as
\begin{equation*}
  r:=\frac\rho{1+\frac{\rho^2}{4b}}\qquad\Rightarrow\qquad
  dr=\frac{1-\frac{\rho^2}{4b}}{\left(1+\frac{\rho^2}{4b}\right)^2}.
\end{equation*}
Then the induced metric takes the conformally Euclidean (pseudo-Euclidean) form
\begin{equation*}
  ds^2=\frac{d\rho^2+\rho^2d\Om}{\left(1+\frac{\rho^2}{4b}\right)^2}.
\end{equation*}
Returning to the Cartesian coordinates
$\lbrace\rho,\chi^1,\dotsc,\chi^{p+q-1}\rbrace\mapsto\lbrace x^\al\rbrace$,
we find metric (\ref{qbcvdy}), where
\begin{equation*}
  R=-\frac2b.
\end{equation*}

The above construction shows that the metric on a constant-curvature space is
locally isometric to that of either the Euclidean (pseudo-Euclidean) space
($R=0$), or to the sphere $\MS^{p+q}$, or to the hyperboloid $\MH^{p+q}$
depending on the metric signature and the sign of the scalar curvature.

Euclidean (pseudo-Euclidean) metric (\ref{qxbvct}) and the hypersurfaces defined
by Eqn (\ref{qbdftr}) are invariant under the rotation group $\MO(p+1,q)$
transformations. Hence,
\begin{equation*}
  \dim\MO(p+1,q)=\frac{n(n+1)}2,\qquad n:=p+q,
\end{equation*}
and, in accordance with Theorem \ref{tkildi}, the number of independent Killing
vectors is maximal and the space of constant curvature is homogeneous and
isotropic.
\section{Symmetric tensors on a constant-curvature space         \label{bvxeem}}
It was shown is Section \ref{sxzlpf} that a homogeneous and isotropic
$n$-dimensional manifold is necessarily a constant-curvature space with
maximal number $n(n+1)/2$ of linearly independent Killing vector fields. Such
spaces are common in applications. Moreover, they can carry other tensor fields,
for example, matter fields in general relativity. In order to have a symmetric
model, it is necessary to impose the symmetry condition on both the metric and
other fields. In this section, we find conditions such that the simplest tensor
fields on a constant-curvature space are also homogeneous and isotropic.

Let
\begin{equation*}
  T=dx^\al\otimes\dotsc\otimes dx^\bt\, T_{\al\dotsc\bt}.
\end{equation*}
be an arbitrary tensor field on a constant-curvature space $\MS$. To be
specific, we consider covariant tensor fields. We assume that an isometry
$\imath:~x\mapsto x'$ is given. Then the requirement that a given tensor field
be symmetric with respect to the isometry group has the same form as for the
metric (\ref{eisoma}):
\begin{equation*}
  T(x)=\imath^*T(x'),
\end{equation*}
where $\imath^*$ is the map of differential forms. This condition has the
component form
\begin{equation}                                                  \label{qmnsuh}
  T_{\al\dotsc\bt}(x)=\frac{\pl x^{\prime\g}}{\pl x^\al}\dotsc
  \frac{\pl x^{\prime\dl}}{\pl x^\bt}T_{\g\dotsc\dl}(x').
\end{equation}
Let an infinitesimal isometry be generated by a Killing field $K=K^\al\pl_\al$.
Then symmetry condition (\ref{qmnsuh}) means that the corresponding Lie
derivative vanishes:
\begin{equation}                                                  \label{qmfkdf}
  \Lie_K T=0.
\end{equation}

The same symmetry condition must be satisfied for any tensor field with both
covariant and contravariant indices.

We now consider the simplest cases common in applications.
\begin{exa}{\rm
Let a differentiable scalar field $\vf(x)\in\CC^1(\MS)$ (function) be given on a
constant-curvature space $\MS$. The vanishing Lie derivative condition then
takes the form
\begin{equation*}
  K^\al(x)\pl_\al\vf(x)=0.
\end{equation*}
An invariant scalar field must be constant, $\vf=\const$, on the whole $\MS$
because the Killing vector field components $K^\al(x)$ can take arbitrary values
at any point $x\in\MS$. Thus, a homogeneous and isotropic field on a
constant-curvature space $\MS$ must be constant: $\vf(x)=\const$ for all
$x\in\MS$.}
\qed\end{exa}
\begin{exa}{\rm
We consider a differentiable covector field $A=dx^\al A_\al$. Then invariance
condition (\ref{qmfkdf}) takes the form
\begin{equation*}
  K^\bt\pl_\bt A_\al+\pl_\al K^\bt A_\bt=0.
\end{equation*}
We choose Killing vectors such that the equality $K^\bt(x)=0$ is satisfied at an
arbitrary but fixed point $x\in\MS$. Moreover, Killing vectors can be chosen
such that the partial derivatives $\pl_\bt K_\al$ are arbitrary and
antisymmetric at a given point. Because $\pl_\al K^\bt=\nb_\al K^\bt$ at a given
point, the equalities
\begin{equation*}
  \pl_\al K^\bt A_\bt=\pl_\al K_\bt A^\bt=\pl_\g K_\bt(\dl^\g_\al A^\bt)
\end{equation*}
hold. It follows that
\begin{equation*}
  \dl_\al^\g A^\bt=\dl_\al^\bt A^\g,
\end{equation*}
because the construction works for any point of $\MS$. Contracting the indices
$\al$ and $\g$ yields the equality
\begin{equation*}
  nA^\bt=A^\bt.
\end{equation*}
Thus, except for the trivial case $n=1$, lowering the index yields $A_\al=0$.
Consequently, if a covector field is homogeneous and isotropic, it vanishes
identically.

The same is true for vector fields $X=X^\al\pl_\al$: a homogeneous and isotropic
vector field on a constant-curvature space $\MS$ necessarily vanishes.}
\qed\end{exa}
\begin{exa}{\rm                                                   \label{eabhhf}
As a third example, we consider a differentiable second-rank covariant tensor
$T_{\al\bt}$. We assume no symmetry in the indices $\al$ and $\bt$. The Lie
derivative of a second-rank tensor is given by
\begin{equation*}
  \Lie_K T_{\al\bt}=K^\g\pl_\g T_{\al\bt}+\pl_\al K^\g T_{\g\bt}
  +\pl_\bt K^\g T_{\al\g}.
\end{equation*}
As in the preceding case, we choose the Killing vector such that the relation
$K^\g(x)=0$ is satisfied at a point $x\in\MS$ and the partial derivatives
$\pl_\al K_\bt$ are antisymmetric. Then, equating the Lie derivative to zero, we
find
\begin{equation*}
  \dl_\al^\dl T^\g{}_\bt+\dl_\bt^\dl T_\al{}^\g=\dl_\al^\g T^\dl{}_\bt
  +\dl_\bt^\g T_\al{}^\dl.
\end{equation*}
Contracting the indices $\al$ and $\dl$ and lowering $\g$ yields
\begin{equation*}
  (n-1)T_{\g\bt}+T_{\bt\g}=g_{\bt\g}T,\qquad T:=T_\al{}^\al.
\end{equation*}
Transposing the indices $\bt$ and $\g$ and subtracting the resulting expression,
we obtain
\begin{equation*}
  (n-2)(T_{\g\bt}-T_{\bt\g})=0.
\end{equation*}

It follows that when $n\ne2$, an invariant second-rank tensor is symmetric.
Using this symmetry, we find that
\begin{equation*}
  T_{\al\bt}=\frac Tng_{\al\bt}.
\end{equation*}
Because the trace $T$ is a scalar, it must be constant by symmetry arguments
from the first example. Therefore, a homogeneous and isotropic second-rank
tensor on a constant-curvature space is given by
\begin{equation}                                                  \label{qmkujh}
  T_{\al\bt}=Cg_{\al\bt},\qquad C=\const.
\end{equation}

This formula holds for $n\ge3$, and for the symmetric part, at $n=2$.

In the two-dimensional case, a homogeneous and isotropic covariant tensor can
have an antisymmetric part proportional to the totally antisymmetric second-rank
tensor $\ve_{\al\bt}=-\ve_{\bt\al}$:
\begin{equation*}
  T_{[\al\bt]}=-T_{[\bt\al]}=C\ve_{\al\bt},
\end{equation*}
if the symmetry under space reflections is disregarded. The sign of the
antisymmetric tensor changes under reflections:
$\ve_{\al\bt}\mapsto-\ve_{\al\bt}$. Therefore, a homogeneous and isotropic
tensor invariant under reflections in two dimensions have the same form
(\ref{qmkujh}) as in the higher-dimensional case.

Homogeneous and isotropic contravariant second-rank tensors and mixed-symmetry
tensors
\begin{equation*}
  T^{\al\bt}=Cg^{\al\bt},\qquad T^\al{}_\bt=C\dl^\al_\bt.
\end{equation*}
can be considered along the same lines. The resulting expressions for
homogeneous and isotropic tensors are used in cosmological models, where
$T_{\al\bt}$ plays the role of the matter stress-energy tensor.}
\qed\end{exa}
\section{Manifolds with maximally symmetric submanifolds         \label{seqway}}
In many physical applications, for example, in cosmology, a Riemannian
(pseudo-Rieman\-nian) manifold $\MM$, $\dim\MM=n$, is a topological product of two
manifolds, $\MM=\MR\times\MS$, where $\MR$ is the real line identified with
the time, and $\MS$ is a constant-curvature space. For any $t\in\MR$, there is a
submanifold $\MS\subset\MM$. Because $\MS$ is a constant-curvature space, it is
homogeneous and isotropic. The corresponding isometry group is generated by
$n(n-1)/2$ Killing vectors $\MS$, where $n:=\dim\MM$. In this section, we find
the most general form of the metric on $\MM$ that is invariant under the
transformation group generated by the isometry group on the submanifold $\MS$.

Let $x^\mu$, $\mu=1,\dotsc,n-1$ be coordinates on the constant-curvature space
$\MS$. Then the metric on $\MS$ is $\overset\circ g_{\mu\nu}(x)$. By
construction, it is invariant under the isometry group generated by the Killing
fields $K_i=K^\mu_i(x)\pl_\mu$, $i=1,\dotsc,n(n-1)/2$.

We suppose that a sufficiently smooth metric $g$ of the Lorentzian signature is
defined on the whole $\MM=\MR\times\MS$, and $t\in\MR$ is the time coordinate,
i.e., $g_{00}>0$. We also suppose that all $t=\const$ sections are space-like.
Moreover, we assume that the restriction of the metric $g$ to $\MS$ coincides
with $\overset\circ g_{\mu\nu}$ for any fixed time. Obviously, such a metric has
the form
\begin{equation}                                                  \label{qawesr}
  g_{\al\bt}=\begin{pmatrix}
  g_{00} & g_{0\nu} \\ g_{\mu0} &
  h_{\mu\nu} \end{pmatrix},
\end{equation}
where $g_{00}(t,x)$ and $g_{0\mu}(t,x)=g_{\mu0}(t,x)$ are arbitrary functions of
$t$ and $x$, and $h_{\mu\nu}(t,x)$ is a constant-curvature metric on $\MS$,
where $t$ is a parameter. All the metric components are supposed to be
sufficiently smooth with respect to both $t$ and $x$. The matrix
\begin{equation*}
  h_{\mu\nu}-\frac{g_{0\mu}g_{0\nu}}{g_{00}}
\end{equation*}
is negative definite, because the metric $g_{\al\bt}$ is of the Lorentzian
signature. Moreover, the matrix $h_{\mu\nu}$ is also negative definite by
construction.

First of all, we continue the action of the isometry group from $\MS$ to the
whole $\MM$ as follows. We suppose that the Killing vector field components
$K_i^\mu(t,x)$ parametrically depend on $t$. We define the action of
infinitesimal isometries on $\MM$ by the relations
\begin{equation}                                                  \label{qnjiis}
\begin{split}
  t&\mapsto t^\prime=t,
\\
  x^\mu&\mapsto x^{\prime\mu}=x^\mu+\e K^\mu,\qquad\e\ll1,
\end{split}
\end{equation}
where $K$ is an arbitrary Killing vector from the Lie algebra generated by the
vectors $K_i$. In other words, the isometry transformations do not shift points
on the real axis $t\in\MR\subset\MM$. This means that Killing vectors are
continued to the whole $\MM$ such that the extra components is absent:
$K^0\pl_0=0$. The continuation is nontrivial if the Killing vector fields become
parametrically dependent on $t$. The resulting Lie algebra of Killing vector
fields continued to $\MM$ is the same.
\begin{exa}{\rm
In four dimensions, Killing vector fields continued to the whole
$\MM=\MR\times\MS$ generate the isometry group $(\MM,\MG)$, where
\begin{equation*}
  \MG=\begin{cases} \MS\MO(4), & \MS=\MS^3\quad\text{-- sphere}, \\
  \MI\MS\MO(3), & \MS=\MR^3\quad \text{-- Euclidean space}, \\
  \MS\MO(3,1), & \MS=\MH^3\quad \text{-- two-sheeted hyperboloid}.
\end{cases}
\end{equation*}
This example is important in cosmology.}
\qed\end{exa}

We can now define a homogeneous and isotropic space-time.
\begin{defn}
A space-time $(\MM,g)$ is called {\em homogeneous and isotropic} if
\newline
(1) \parbox[t]{.95\linewidth}{the manifold is the topological product
$\MM=\MR\times\MS$, where $\MR$ is the time axis and $\MS$ is a
three-dimensional constant-curvature space endowed with negative-definite
metric;}\newline
(2)  \parbox[t]{.95\linewidth}{the metric $g$ is invariant under transformations
(\ref{qnjiis})generated by the isometry group of $\MS$.\qed}
\end{defn}

We find the most general form of a homogeneous and isotropic metric of the
universe.
\begin{theorem}                                                   \label{tnbdht}
Let metric (\ref{qawesr}) on $\MM=\MR\times\MS$ be sufficiently smooth and
invariant under transformations (\ref{qnjiis}). Then, in some neighborhood of
any point, a coordinate system exists such that the metric is block-diagonal
\begin{equation}                                                  \label{qvsear}
  ds^2=dt^2+h_{\mu\nu}dx^\mu dx^\nu,
\end{equation}
where $h_{\mu\nu}(t,x)$ is a constant-curvature metric on $\MS$ for all
$t\in\MR$. Moreover, the Killing vector field components are independent of
time.
\end{theorem}
\begin{proof}
Let $x^\mu$ be coordinates on $\MS$. We fix one of the hypersurfaces $t=\const$.
The corresponding tangent vector has spatial components only: $X=X^\mu\pl_\mu$.
The corresponding orthogonal vector $n^\al\pl_\al$ must satisfy the relation
\begin{equation*}
  n^0X^\nu g_{0\nu}+n^\mu X^\nu g_{\mu\nu}=0.
\end{equation*}
This equality must be satisfied for all tangent vectors $X$, thereby giving rise
to spatial components of normal vectors
\begin{equation*}
  n^\mu=-n^0g_{0\nu}\hat g^{\mu\nu},
\end{equation*}
where $\hat g^{\mu\nu}$ is the inverse spatial metric,
$\hat g^{\mu\nu}g_{\nu\rho}=\dl^\mu_\rho$. It is easy to show that normal
vectors are time-like.

We now draw a geodesic line through each point of the space-like hypersurface
$x\in\MS$ along the normal direction. We choose the geodesic length $s$ as the
time coordinate. Without loss of generality, we can assume that the initial
space-like hypersurface corresponds to $s=0$. Thus, we have built a coordinate
system $\lbrace x^\al\rbrace=\lbrace x^0:=s,x^\mu\rbrace$ in some neighborhood
of the hypersurface $\MS$.

By construction, the lines $x^\al(\tau)$ of the form
$\lbrace x^0=s,x^\mu=\const\rbrace$, where $\tau:=s$ are geodesics with the
velocity vector $\dot x^\al=\dl^\al_0$. From the geodesic equation
\begin{equation*}
  \ddot x^\al=-\Gamma_{\bt\g}{}^\al\dot x^\bt\dot x^\g,
\end{equation*}
it follows that $\Gamma_{00}{}^\al=0$ in the coordinate system under
consideration. Lowering the index $\al$, we find an equation for the metric
components:
\begin{equation}                                                  \label{eqcote}
  \pl_0 g_{0\al}-\frac12\pl_\al g_{00}=0.
\end{equation}
By construction, the time-like tangent vector $\pl_0$ has a unit length. It
follows that $g_{00}=1$. Then Eqn (\ref{eqcote}) takes the form
$\pl_0g_{0\mu}=0$. This differential equation can be solved with the initial
condition $g_{0\mu}(s=0)=0$, because the vector $n$ is perpendicular to the
initial hypersurface. For differentiable functions $g_{0\mu}$, the equation has
unique solution $g_{0\mu}=0$. Thus, the metric is of block-diagonal form
(\ref{qvsear}) in the resulting coordinate system.

The hypersurfaces $t=\const$ given above are called {\em geodesically parallel}.

So far, we have ignored the properties of constant-curvature surfaces. The proof
is general and implies that locally there exists a `temporal' gauge for the
metric (or the synchronous coordinate system).

By construction, the zeroth component of the Killing vector vanishes,
$K^0(0,x)=0$, on the hypersurface $s=0$. The $(0,0)$-component of the Killing
equations, which can be more conveniently written in form (\ref{ekilco}), yields
the equation $\pl_s K^0(s,x)=0$. For sufficiently smooth functions, this
equation with initial condition $K^0(0,x)=0$ has a unique solution,
$K^0(s,x)=0$, for all $s$ admissible in the coordinate system. As a result, all
hypersurfaces $s=\const$ in some neighborhood of the initial hypersurface have
constant curvature.

If the metric (\ref{qvsear}) is block-diagonal, then the $(0,\mu)$-components
of Killing equations (\ref{ekilco}) take the form $\pl_s K^\mu=0$. It follows
that the Killing vector field is independent of time.

The spatial $(\mu,\nu)$-components of the Killing equations are satisfied
because $K$ is the Killing vector field on $\MS$.

Returning to the original notation $s\mapsto t$, we obtain metric
(\ref{qvsear}).
\end{proof}

Hilbert introduced coordinates in which the metric takes block-diagonal form
(\ref{qvsear}) (see Eqn (22) in his paper \cite{Hilber24}). The resulting
coordinate system was called Gaussian. However, the corresponding spatial
sections were not constant-curvature spaces, and Killing vector fields were not
considered.

If the metric is block-diagonal, Eqn (\ref{qvsear}), and $K=K^\mu\pl_\mu$, then
Killing equations (\ref{ekilco}) split into temporal and spatial components:
\begin{alignat}{2}                                                \label{qnjvyt}
  (\al,\bt)&=(0,0):\qquad\qquad & 0&=0,
\\                                                                \label{qnbcks}
  (\al,\bt)&=(0,\mu): & h_{\mu\nu}\pl_0 K^\nu&=0,
\\                                                                \label{qawqjh}
  (\al,\bt)&=(\mu,\nu): & h_{\mu\rho}\pl_\nu K^\rho+h_{\nu\rho}\pl_\mu K^\rho
  +K^\rho\pl_\rho h_{\mu\nu}&=0.
\end{alignat}

\begin{theorem}                                                   \label{tsrewh}
Under the assumptions of Theorem \ref{tnbdht} metric (\ref{qvsear}) has the form
\begin{equation}                                                  \label{qbvxfo}
  ds^2=dt^2+a^2\overset\circ g_{\mu\nu}dx^\mu dx^\nu,
\end{equation}
where $a(t)>0$ is an arbitrary sufficiently smooth function (the scale factor)
and $\overset\circ g_{\mu\nu}(x)$ is a constant-curvature metric depending only
on spatial coordinates $x\in\MS$.
\end{theorem}
\begin{proof}
The Killing equations (\ref{qawqjh}) are satisfied because $h_{\mu\nu}(t,x)$ is
a constant-curvature metric on $\MS$ for all $t\in\MR$. Theorem \ref{tnbdht}
claims that Killing vector fields are independent of time. Thus, differentiating
equation (\ref{qawqjh}) in time, we find the relation
\begin{equation*}
  \dot h_{\mu\rho}\pl_\nu K^\rho+\dot h_{\nu\rho}\pl_\mu K^\rho
  +K^\rho\pl_\rho \dot h_{\mu\nu}=0.
\end{equation*}
This implies that the time derivative of the metric $\dot h_{\mu\nu}$ is a
homogeneous and isotropic second-rank tensor. Example \ref{eabhhf} says that the
time derivative must be proportional to the metric itself:
\begin{equation}                                                  \label{qbncgt}
  \dot h_{\mu\nu}=f h_{\mu\nu},
\end{equation}
where $f(t)$ is a sufficiently smooth function of time.

If $f=0$, the proof is trivial, and the metric is already of form (\ref{qbvxfo})
for $a=\const$.

Letting $f\ne0$, we introduce a new temporal coordinate $t\mapsto t'$ defined by
the differential equation
\begin{equation*}
  dt'=f(t)dt.
\end{equation*}
Equation (\ref{qbncgt}) then takes the form
\begin{equation*}
  \frac{dh_{\mu\nu}}{dt'}=h_{\mu\nu}.
\end{equation*}
The general solution is given by
\begin{equation*}
  h_{\mu\nu}(t',x)=C\exp(t')\overset\circ g_{\mu\nu}(x),\qquad C=\const\ne0,
\end{equation*}
where $\overset\circ g_{\mu\nu}(x)$ is a constant-curvature metric on $\MS$,
which is independent of time. Hence, representation (\ref{qbvxfo}) follows.
\end{proof}
Theorem \ref{tbchyt} follows from theorems \ref{tnbdht} and \ref{tsrewh}.
\section{Example                                                 \label{sjhdtr}}
The explicit form of the Friedmann metric for a homogeneous and isotropic
universe, Eqn (\ref{qbvxfo}), depends on coordinates on the constant-curvature
space. The Friedmann metric in the stereographic coordinates is diagonal:
\begin{equation}                                                  \label{emfree}
  g=\begin{pmatrix} 1 & 0 \\
  0 & \frac{\displaystyle a^2\eta_{\mu\nu}}{\displaystyle \big(1+b_0x^2\big)^2}
  \end{pmatrix},
\end{equation}
where $b_0=-1,0,1$, $\eta_{\mu\nu}:=\diag(---)$ is the negative-definite
Euclidean metric and $x^2:=\eta_{\mu\nu}x^\mu x^\nu\le0$. Because the metric on
spatial sections is negative definite, the values $b_0=-1,0,1$ correspond to the
respective spaces of negative, zero and positive curvature. In cases of positive
and zero curvature, the stereographic coordinates are defined on the whole
Euclidean space $x\in\MR^3$. In the negative-curvature case, the stereographic
coordinates are defined in the interior of the ball $|x^2|<1/b_0$.

We transform the coordinates as $x^\mu\mapsto x^\mu/a$. The resulting metric
takes the nondiagonal form, while the conformal factor disappears:
\begin{equation}                                                  \label{emfrei}
  g=\begin{pmatrix}
  1+{\displaystyle\frac{\dot b^2x^2}{4b^2\big(1+bx^2\big)^2}} &
  {\displaystyle\frac{\dot bx_\nu}{2b\big(1+bx^2\big)^2}} \\[2mm]
  {\displaystyle\frac{\dot bx_\mu}{2b\big(1+bx^2\big)^2}} &
  {\displaystyle\frac{\eta_{\mu\nu}}{\big(1+bx^2\big)^2}} \end{pmatrix},
\end{equation}
where
\begin{equation}                                                  \label{qbgsty}
  b(t):=\frac{b_0}{a^2(t)},
\end{equation}
and the dot denotes the time derivative.

We see that the metric of a homogeneous and isotropic universe can be
nondiagonal without the conformal factor. Moreover, the scalar curvature of
spatial sections, which is proportional to $b(t)$, explicitly depends on time.

Now we simply disregard off-diagonal elements, choose $g_{00}=1$, and add the
scale factor. Then the metric takes the form
\begin{equation}                                                  \label{emfrep}
  g=\begin{pmatrix} 1 & 0 \\
  0 & {\displaystyle\frac{a^2\eta_{\mu\nu}}{\big(1+bx^2\big)^2}} \end{pmatrix}.
\end{equation}
This metric contains two arbitrary independent functions of time: $a(t)>0$ and
$b(t)$. It is nondegenerate for any $b$, including zero. All $t=\const$ sections
of the corresponding space-time are obviously spaces of constant curvature and
are therefore homogeneous and isotropic. The metric is interesting because, in
general, it can be used to analyze solutions passing through the zero $b=0$. If
such solutions exist, the spatial sections change the curvature from positive to
negative values during the time evolution and vice versa.

We cannot eliminate an arbitrary function $b(t)$ by means of a coordinate
transformation without producing off-diagonal terms.

There is a curious situation. On one hand, all spatial sections of metric
(\ref{emfrep}) are homogeneous and isotropic. On the other hand, any
homogeneous and isotropic metric must have form (\ref{qkdtwu}). The key is that
(\ref{emfrep}) is in general not homogeneous and isotropic. Indeed, each
$t=\const$ section of the space-time $\MM$ is a constant curvature space, and
spatial $(\mu,\nu)$-components of Killing equations (\ref{qawqjh}) are
satisfied, while mixed $(0,\mu)$-components are not. In the stereographic
coordinates, six independent Killing vectors of the spatial sections are
expressed as
\begin{equation}                                                  \label{qbsghu}
\begin{split}
  \hat K_{0\mu}&=(1+bx^2)\pl_\mu-\frac2bx_\mu x^\nu\pl_\nu,
\\
  \hat K_{\mu\nu}&=x_\mu\pl_\nu-x_\nu\pl_\mu,
\end{split}
\end{equation}
where the indices $\mu,\nu=1,2,3$ label Killing vector fields. The first three
Killing vectors generate translations at the coordinate origin $x^2=0$, while
the last three generate rotations. We see that the first three Killing vector
fields explicitly depend on time through function $b(t)$, while equations
(\ref{qnbcks}) are not satisfied.

There is another method to see that metric (\ref{emfrep}) is not homogeneous
and isotropic. Direct calculation yields the scalar curvature:
\begin{equation*}
  R=-\frac{24b}{a^2}+6\left[\frac{\ddot a}a+\frac{\dot a^2}{a^2}
  -\frac1{1+bx^2}\left(4\frac{\dot a\dot bx^2}a+\ddot bx^2\right)
  +3\frac{\dot b^2x^4}{(1+bx^2)^2}\right],
\end{equation*}
which explicitly depends on $x$ and is therefore not homogeneous and isotropic.

This example shows that the homogeneity and isotropy of spatial sections are not
sufficient for the complete four dimensional metric to be homogeneous and
isotropic. The equivalent definition is as follows.
\begin{defn}
A space-time is called {\em homogeneous and isotropic} if: \newline
(1) \parbox[t]{.95\linewidth}{all constant-time $t=\const$ sections are
constant-curvature spaces $\MS$;}\newline
(2) \parbox[t]{.95\linewidth}{the extrinsic curvature of hypersurfaces
$\MS\hookrightarrow\MM$ is homogeneous and isotropic.\qed}
\end{defn}

The definition of the extrinsic curvature of embedded surfaces can be found,
i.e., in \cite{Wald84,Katana13B}. In our notation, the extrinsic curvature
$K_{\mu\nu}$ for block-diagonal metric (\ref{qvsear}) is proportional to the
time derivative of the spatial part of the metric
\begin{equation*}
  K_{\mu\nu}=-\frac12\dot h_{\mu\nu}.
\end{equation*}

The last definition of a homogeneous and isotropic space-time is equivalent to
the definition given in Section \ref{seqway}. Indeed, the first condition
implies that the space-time is a topological product $\MM=\MR\times\MS$. It
follows that the metric can be mapped into block-diagonal form (\ref{qvsear}).
Then the second condition in the definition yields Eqn (\ref{qbncgt}), and we
can follow the proof of Theorem \ref{tsrewh}.

We note that the second condition in the definition of a homogeneous and
isotropic universe is necessary because metric (\ref{emfrep}) provides a
counterexample.
\section{Conclusion}
In this paper, we have given two equivalent definitions of a homogeneous and
isotropic space-time. We also explicitly proved Theorem \ref{tbchyt}, which
describes the most general form of a homogeneous and isotropic metric up to a
coordinate transformation. This is the Friedmann metric. Although the theorem is
known, its proof and corresponding definitions are difficult to find in the
literature. It seems that the proof of Theorem \ref{tsrewh} and the second
definition of a homogeneous and isotropic space-time were not known before. The
proof of Theorem \ref{tsrewh} is simple, but not simpler than the one in book
\cite{Weinbe72}. However, it is well adapted to proving the equivalence of the
definitions.

The research was supported by the Russian Science Foundation
(Project ¹14-50-00005).

\end{document}